\documentclass[12pt, oneside, reqno]{article}
 \usepackage{graphicx}
\usepackage{amsmath,amsthm,amsfonts,amscd,amssymb,mathrsfs}
\usepackage{authblk}
\usepackage{enumerate}
\usepackage[all]{xypic}

\usepackage{amscd}   % for commutative diagrams
\usepackage[all]{xy} % for complicated commutative diagrams
\usepackage{booktabs} % for nicer tables
\usepackage{bbm}
\usepackage{physics}
\usepackage{hyperref}
\hypersetup{colorlinks=true,citecolor=blue, urlcolor= cyan}
\usepackage{qcircuit}
\usepackage{tikz}

\usepackage{multirow}

\usepackage{ytableau}
\usepackage[enableskew]{youngtab} % For young tableaux

\usepackage{nicefrac}
\usepackage{xfrac}

\usepackage{adjustbox}
\usepackage{booktabs, makecell, multirow}
\usepackage{geometry}
\usepackage{array}
\usepackage{tabularx}
\usepackage{subcaption}

\usepackage{color}

\makeatletter
\newtheorem{rep@theorem}{\rep@title}
\newcommand{\newreptheorem}[2]{%
\newenvironment{rep#1}[1]{%
 \def\rep@title{#2 \ref{##1}}%
 \begin{rep@theorem}}%
 {\end{rep@theorem}}}
\makeatother

\newtheorem{theorem}{Theorem}[section]
\newreptheorem{theorem}{Theorem}
\newreptheorem{lemma}{Lemma}

\newtheorem{lemma}[theorem]{Lemma}

\newtheorem{proposition}[theorem]{Proposition}

\def\phi{ \varphi }

\theoremstyle{definition}
\newtheorem{definition}[theorem]{Definition}
\newtheorem{example}[theorem]{Example}

\theoremstyle{remark}

\newcommand{\ccirc}[1]{\xymatrix@1{+<1ex>[o][F-]{#1}}}

\usepackage[backend=biber,
style=numeric,
sorting=none
]{biblatex}

\title{Exploiting Finite Geometries for Better Quantum Advantages in Mermin-Like Games}
\author[$\dagger$]{Colm Kelleher}
\author[$\dagger,\ddagger$]{Frédéric Holweck}
\author[$\star$]{P\'eter L\'evay}

\affil[$\dagger$]{Laboratoire Interdisciplinaire Carnot de Bourgogne, ICB/UTBM, UMR 6303 CNRS,Universit\'e de Technologie de Belfort-Montbéliard, 90010 Belfort Cedex, France}
\affil[$\ddagger $]{Department of Mathematics and Statistics, Auburn University, Auburn, AL, USA}
\affil[$\star$]{MTA-BME Quantum Dynamics and Correlations Research Group, Budapest University of Technology and Economics, Műegyetem rkp. 3., H-1111 Budapest Hungary}

\begin{document}

\maketitle
% \tableofcontents

\begin{abstract}
Quantum games embody non-intuitive consequences of quantum phenomena, such as entanglement and contextuality. The Mermin-Peres game is a simple example, demonstrating how two players can utilise shared quantum information to win a no - communication game with certainty, where classical players cannot. In this paper we look at the geometric structure behind such classical strategies, and borrow ideas from the geometry of symplectic polar spaces to maximise this quantum advantage. We introduce a new game called the Eloily game with a quantum-classical success gap of $0.2\overline{6}$, larger than that of the Mermin-Peres and doily games. We simulate this game in the IBM Quantum Experience and obtain a success rate of $1$, beating the classical bound of $0.7\overline{3}$ demonstrating the efficiency of the quantum strategy. 
\end{abstract}

\section{Introduction}
Within the realm of quantum mechanics, resources such as superposition and entanglement have been exhaustively studied as useful for both post-classical conceptualisation and technological tools. There is another quantum resource - contextuality - that is often overlooked compared to the other two. It has enjoyed recent attention as a useful resource for quantum computation, see \cite{Budroni21} for a survey on the topic.
\\~\\
Contextuality relates the measurement outcomes of quantum operators
with the set of mutually commuting operators that they exist within. In classical physics, measurement outcomes are independent of the set of measurements also taken on the system before and after the measurement in question. In quantum mechanics, as measurement operators do not in general commute, the set of operators that do commute with the operator in question can be understood to have an affect on the measurement outcome. A set of mutually commuting operators, such that their product is equal to $\pm I_d$, is known as a \textit{context}.
\\~\\
In 1964 John Bell first introduced the notion of non-locality as a necessary condition on hidden quantum variables \cite{bell1964einstein}. Two years later he published a lesser well-known result on hidden variables and contextuality \cite{bell_problem_1966}. This coincided with a paper by Kochen and Specker with a similar result \cite{Kochen_specker}, which is more well known by physicists as the Kochen-Specker (KS) Theorem. This theorem states that any hidden variable model that reproduces the outcomes of quantum mechanics should be context-dependent. Thirty years later, Mermin \cite{Mermin93} and Peres \cite{Peres90} provided an observable-based proof of this theorem, see Section \ref{sec:quantum_games}. In 2002 Aravind reformulated this theorem in terms of a two player, no-communication game with advantage given by quantum information \cite{aravind_bells_2002, brassard_quantum_2005}.
\\~\\
In this work we introduce another such game associated to the Bell-KS Theorem. This new game has a maximal gap between predicted classical outcome and outcomes given by quantum mechanics (i.e. no non-contextual hidden variable models). In Sections \ref{sec:quantum_games}, \ref{sec:doily} we outline the background for quantum games, and how to generalise them based on finite geometries. In Section \ref{sec:symplectic_polar_space} we introduce the notion of symplectic polar spaces, which allow us to construct the canonical ``Eloily" arrangement, furnishing our new game. Section \ref{sec:gq_24_game} describes the associated game, its construction using stabiliser states, and its implementation in the IBM Quantum Experience. We present the results and conclusions in Section \ref{sec:results}. Finally, in Section \ref{sec:quadrangles} we discuss the common formalism of all games presented here and an invariant connecting it with the well-known Cabello inequality\cite{cabello_proposed_2010}.

\subsection{Quantum Games - the Mermin Perez Magic Square Game}\label{sec:quantum_games}

The Mermin-Peres Magic Square is a $3 \times 3$ geometric arrangement of quantum spin operators, demonstrating the impossibility of non-contextual hidden variable models and proving the Bell-KS theorem \cite{bell_problem_1966, Kochen_specker}. An example of the titular square is given in Fig \ref{fig:merminsquare}, with tensor products omitted between single-qubit Pauli operators. Multiplying the 2-qubit operators along any vertical or horizontal line (context) produces the identity $I_{4}$ scaled by $\pm$1. 
\begin{figure}[!h]
\centering
 \includegraphics[width=4cm]{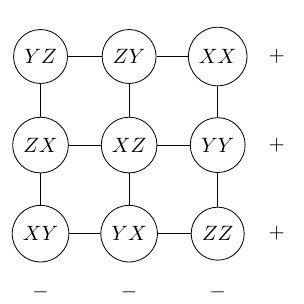}
 \caption{An example of Mermin-Peres Magic Square. Lines denote contexts, which are composed of mutually-commuting operators. Vertical contexts multiply to give $-I_{4}$, while horizontal contexts multiply to give $+I_{4}$. The eigenvalues of the operators multiply to give the eigenvalues of the product - for instance the eigenvalues of the top horizontal context must multiply to $+1$.}\label{fig:merminsquare}
\end{figure}
As each context is composed of mutually commuting operators, the eigenvalues of each operator multiply together to give the eigenvalues of the product of operators. As the products are given by identity up to sign, this sets a simple condition on the individual eigenvalues of any given line/context. It is easy to check that no assignment of $\pm1$ eigenvalues can be made to the points in the square such that the context constraints are satisfied unless one decides that the assignments are context-dependent. This demonstrates the impossibility of any non-contextual hidden variable (NCHV) theory \cite{Mermin93, Peres90}.
\\~\\
This can be converted into a quantum game as follows \cite{aravind_bells_2002, Aravind04}:\\
A referee Charlie provides two players - Alice and Bob - each with a number in the set $\{1,2,3\}$, see Fig. \ref{fig:mermin_game}. Alice and Bob cooperate to win together but cannot communicate during the game and so do not know each other's numbers. They respond to Charlie with strings such that Alice's multiply to $+1$, and Bob's to $-1$. They win iff Alice's response in the position of Bob's number is equal to Bob's response in the position of Alice's number.\\

\begin{figure}
\centering
    \includegraphics[width=0.3\textwidth]{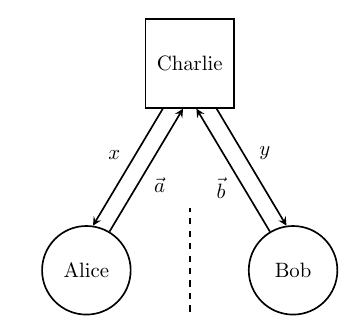}
    \caption{The Mermin-Peres game. Charlie gives Alice and Bob a number $x,y$, who in response return answers $\vec{a} = (a_{1},a_{2},a_{3})$, $\vec{b} = (b_{1},b_{2},b_{3})$. The $a_{i}$ and $b_{i}$ are all selected from the set $\{+1, -1\}$ and conditions on the answers are that $a_{1}a_{2}a_{3} = +1$, $b_{1}b_{2}b_{3} = -1$. The players win iff $a_{y} = b_{x}$.}
    \label{fig:mermin_game}
\end{figure}
They are allowed to cooperate before the game begins, and so share a table of possible responses depending on the numbers given by Charlie. Such a table is given in Figure \ref{fig:mermin_classical}. However, due to the constraints on the rows and columns, it is impossible to fully populate this table in a consistent manner! As the winning condition in the game now relates to the consistency of the intersection points in the table, Alice and Bob can only win with this strategy 8 times out of 9.

\begin{figure}[!h]
\centering
 \includegraphics[width=0.3\textwidth]{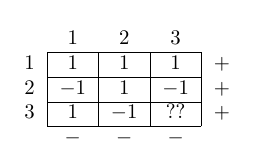}
 \caption{An example of a classical strategy in a $3\times 3$ grid: The triplet sent back by Alice corresponds to the row for a given $x\in\{1,2,3\}$ while the triplet sent back by Bob corresponds to the column $y\in\{1,2,3\}$. The bottom-right cell cannot be populated while keeping both row and column sign conditions met. With this strategy they win the game in $8$ cases out of $9$.}\label{fig:mermin_classical}
\end{figure}

One can however replace the classical table in Fig. \ref{fig:mermin_classical} with a Mermin-Peres square, such as Fig. \ref{fig:merminsquare}. Now Alice and Bob also share a pair of entangled particles before the game begins, and as before depending on the $x,y$ values provided choose a given row \& column. Here, they make their measurements on the particles corresponding to the operators in the table, and note the measurement results as their responses. As the eigenvalues of each operator is either $\pm1$, and the contexts denote mutually commuting operators, the response constraints will be automatically satisfied. 

This approach - known as the quantum strategy - will work with probability 1, beating out the classical strategy. This is a slick demonstration of the Bell-KS theorem, showing that you cannot deterministically apply hidden variables (pre-measurement operator values) to certain operator arrangements in a consistent and context-independent manner. The implementation of this game on noisy quantum computers has been shown in \cite{Holweck21, Kelleher_2_qubit_games, xu_experimental_2022, DRLB}. 

The reason the Mermin-Peres game shows a quantum advantage is due to three factors:
\begin{enumerate}
    \item The Mermin-Peres Grid (Fig. \ref{fig:merminsquare}) provides the players with triplets that automatically satisfy the constraints of the game.
    \item The shared entangled state the players measure on guarantees triplets in agreement on their intersection, winning the game with probability 1.
    \item The \textit{contextual} nature of the grid, as previously stated, implies that there is no NCHV model that can win the game with probability 1.
\end{enumerate}

The degree to which one cannot satisfy all line constraints by classical values is known as the \textit{degree of contextuality} of an arrangement \cite{henri_contextuality_2022}. The degree of contextuality of the Mermin-Peres grid is $1$, meaning there is always one context constraint that is not satisfied by NCHV.

The purpose of this article is to consider different quantum games and compare their quantum advantages:
\begin{definition}
    The \textit{quantum advantage} or quantum-classical success gap is the difference in probability of success between a quantum and classical strategy for a given game.
\end{definition}
For all games here, the quantum success probability is $1$. The advantage depends directly on the degree of contextuality of the associated arrangement. In the Mermin-Peres game, as the classical strategy gives a winning probability of $\omega({\mathscr{M}}) = \frac{8}{9}$, the quantum advantage is $1-\omega({\mathscr{M}}) = \frac{1}{9}$.

A key characteristic of this game is the geometry and constraint structure of the Mermin-Peres Magic Square. This geometry, as well as other arrangements, have been well studied\cite{Brassard05, Arkhipov12, LS}. In this paper, we show how another special geometry - $GQ(2,4)$ - can provide a quantum game with an even larger quantum advantage. But first, we show how to generalise the Mermin-Peres game.

\subsection{The Doily Game}\label{sec:doily}

In the Mermin-Peres square, we had nine 2-qubit operators and three negative lines. There is a single ``constraint" line, in the sense that there is one line that cannot be satisfied in the classically-populated table in Fig. \ref{fig:mermin_classical}. This can be related to the negative lines, as one can check that we can convert two negative lines into positive ones without changing the overall line constraints, by multiplying $XY$ and $YX$ by $-1$. This gives a grid of 1 negative and 5 positive lines with one inconsistent constraint associated to the degree of contextuality. This is the minimum number of negative lines possible in such a grid, and provides us with our single constraint line in Fig. \ref{fig:mermin_classical}.

There are a total of 15 nontrivial 2-qubit operators composed of the Pauli spin operators $\{I, X, Y, Z\}$. Here we exclude the trivial operator $II$ and ignore overall phase factors. One can construct a general geometry from all 15 operators, which contains a total of 3 negative lines, and 3 unsatisfiable constraints that can be mapped to the negative lines \cite{henri_contextuality_2022}. Such a geometry is known as the Doily, and is shown in Fig \ref{fig:doily}. It has 15 points connected by 15 lines, each with 3 points. Each point is contained in 3 lines. Its contextuality has been tested on a quantum computer in \cite{Holweck21}.

\begin{figure}[!h]
 \begin{center}
  \includegraphics[width=5cm]{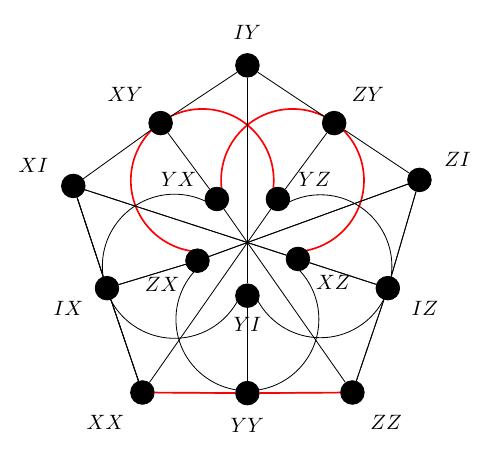}
  \caption{The doily: A $15_3$ point-line configuration that encodes the commutation relations of the two-qubit Pauli group. The red lines are such that the product of the observable is $-I_{4}$ while it is $+I_{4}$ for the other lines. It has degree of contextuality  3.}\label{fig:doily}
 \end{center}
\end{figure}

The doily configuration has been of special interest in multi-qubit arrangements\cite{SZ, muller_multi_qubit_2022, saniga_veldkamp_GQ_2010, henri_contextuality_2022, kelleher_x_states_2021}, in particular its relevance in black hole-qubit correspondence (see for example \cite{saniga_doily_gem}).

Here again we can construct a game based on the Mermin-Peres arrangement (see \cite{Kelleher_2_qubit_games} for more details). Suppose Charlie has an incidence structure of the doily with negative line structure as in Fig. \ref{fig:doily}. We also assume the lines and points of this doily are labelled. Charlie chooses a given point, and gives Alice and Bob distinct lines intersecting at that point. Alice and Bob provide triplets of $\pm1$ satisfying the positive/negative conditions of their lines. They win the game iff the triplet values agree on the chosen intersecting point.
\\~\\
As in the Mermin-Peres game, there is a quantum strategy formed by populating points with 2-qubit operators as in Fig. \ref{fig:doily}, and measuring on a shared entangled state. Like in the grid, this strategy wins with probability 1.
\\~\\
Due to the structure of the doily geometry, the classical success probability $\omega({\mathscr{D}})$ can be determined by analysing the best-case scenario for success. 
The degree of contextuality for the doily is 3, meaning there are always 3 lines that cannot be satisfied with any NCHV strategy. These lines can be associated with the negative lines \cite{henri_contextuality_2022}, and thus do not intersect. For a failure to occur, Charlie must choose one of the 9 points on these lines, and provide that line to one player. The other will necessarily get one of two satisfiable lines, also intersecting at that point. Then there is a $\frac{1}{3}$ chance that the inconsistency on the constraint line is at the chosen intersection point.
\\~\\
Thus, the probability of success for any classical strategy on the doily is

\begin{equation}\label{eq:doily_classical_bound}
 \omega({\mathscr{D}})=1 - \frac{9}{15}\times\frac{2}{3}\times\frac{1}{3}=\frac{13}{15}\approx 0.8667
\end{equation}

This provides a quantum advantage of $\frac{2}{15}$, larger than for the Mermin-Peres game. For this reason, the doily game has also been implemented on a noisy quantum computer using multiple methods \cite{Kelleher_2_qubit_games}.

The above shows that by analysing differing geometries with constraints arising from operator assignments, one can produce a quantum game with a more robust advantage versus the classical counterpart. In Section \ref{sec:gq(2,4)} we find a game with a maximal advantage for geometries with three points per line.
\\~\\
The previous discussion leads to the following lemma, that will be used in Section \ref{sec:gq_24_game}:
\begin{lemma} For a game corresponding to an arrangement of 3 points per line, a situation where one player has a satisfiable context, and the other an unsatisfiable context\footnote{Here ``unsatisfiable context" means a context that cannot be satisfied by a given classical assignment of the points.} results in at best a 2 in 3 chance of winning using a classical strategy.
\end{lemma}\label{lemma:neg_line_two_thirds}
\begin{proof}
    We begin by noting that Alice and Bob are not aware of each other's lines, but coordinate based on a classical strategy that agrees for all satisfiable contexts only. Suppose without loss of generality that Charlie gives Alice a satisfiable line A and Bob an unsatisfiable line B. Alice records her answers based on the pre-agreed strategy. Bob, in order to obtain the line parity condition of B, provides an answer that disagrees with the strategy in at least one position. There is then at least a 1 in 3 chance that this disagreement happens on the intersection point with A.
\end{proof}

\section{Symplectic Polar Space}\label{sec:symplectic_polar_space}

The geometrical structures under discussion here can in general be generated from projective geometrical methods, see in particular \cite{SPPH} and \cite{LHS} for an in-depth analysis of the 2- and 3-qubit cases. 

In particular, here we look at arrangements of 3-qubit operators. The full set of $N$-qubit operators and contexts form the arrangement known as the symplectic polar space $\mathcal{W}(2N-1,2)$, containing $2^{2N}-1$ points. In particular, the doily is the set $\mathcal{W}(3,2)$, with 15 points and 15 lines. 
For 3 qubits, one has $\mathcal{W}(5,2)$, containing 63 points/observables and 315 lines/contexts, each context containing 3 points. Every point can be identified with a 3-qubit operator of the form
\begin{equation}
    \mathcal{O}_{i} = sX^{\mu_{1}}Z^{\nu_{1}}\otimes X^{\mu_{2}}Z^{\nu_{2}}\otimes X^{\mu_{3}}Z^{\nu_{3}}
\end{equation}
where the $\mu_{i},\nu_{i}$ are over the finite field of two elements $\mathbb{F}_{2}$ and $s$ is a phase factor in the set $\{1,-1,i,-i\}$.

This allows us to map non-trivial operators in the space $\mathcal{P}$ to vectors in the projective space $PG(5,2)$ (up to phase factors):

\begin{equation}
\begin{array}{rl}
    \mathcal{P} & \rightarrow PG(5,2) \\
    \mathcal{O}_{p} & \mapsto p = [\mu_{1}:\mu_{2}:\mu_{3}:\nu_{1}:\nu_{2}:\nu_{3}]
\end{array}
\end{equation}

The commutation relations of operators along lines is equivalent to the vanishing of the following symplectic form defined over $\mathbb{F}_{2}$:

\begin{equation}
    \langle p, q \rangle := \sum_{i=1}^3 \mu_{i}\nu_{i}' + \nu_{i}\mu_{i}'
\end{equation}
for $p \in PG(5,2)$ without primed components and $q \in PG(5,2)$ with primed components. In other words, $\mathcal{W}(5,2)$ is nothing more than $PG(5,2)$ equipped with this symplectic form.
\\
One also has the following quadric on $\mathcal{W}(5,2)$:
\begin{equation}
    Q_{0}(p) := \sum_{i=1}^{3}\mu_{i}\nu_{i}
\end{equation}
For a given $p \in \mathcal{W}(5,2)$, we say that the corresponding operator $\mathcal{O}_{p}$ is \textit{skew} if it has an odd number of $Y$ matrices in its tensor product, and \textit{even} otherwise. The condition on skew- and evenness is given by the quadric:
\begin{equation}
Q_{0}(p) = 
\left\{
    \begin{array}{ll}
        0 & \text{ $p$ even}\\
        1 & \text{ $p$ skew}
    \end{array}
\right.
\end{equation}
One can easily see that
\begin{equation}
    \langle p, q \rangle = Q_{0}(p+q) + Q_{0}(p) + Q_{0}(q)
\end{equation}
And one can finally define the quadratic form
\begin{equation}
    Q_{q}(p) := Q_{0}(p) + \langle p, q \rangle^2
\end{equation}
where, as we are over $\mathbb{F}_{2}$, we can omit the square.
\\
Within $\mathcal{W}(5,2)$, one can identify different quadrics of interest, namely the \textit{elliptic} and \textit{hyperbolic} quadrics $E_{q}$ and $H_{q}$, respectively:
\begin{equation}
\begin{array}{rl}
        E_{q} & := \{p \in \mathcal{W}(5,2), Q_{0}(q) = 1, Q_{q}(p) = 0\} \\
        H_{q} & := \{p \in \mathcal{W}(5,2), Q_{0}(q) = 0, Q_{q}(p) = 0\} 
\end{array}
\end{equation}
In both quadrics $p$ are the points that are even and commute with $q$, or are skew and anti-commute with $q$. $E_{q}$ is parameterised by a skew element, $H_{q}$ by an even one. There are 28 elliptic quadrics and 35 + 1 hyperbolics, as one can choose the point $q=III$ for $H_{q}$. 

\subsection{The Eloily: The Elliptic Quadric $E_{q}$}\label{sec:gq(2,4)}

To calculate a specific example, let us consider $E_{YYY}$ and $H_{III}$, which form the ``canonical" labelling of these quadrics. The intersection of these form a doily $D_{YYY}$ made of the 15 symmetric elements that commute with $YYY$.
\\~\\
These quadrics $E_{YYY}$, $H_{III}$ form part of a structure known as the Magic Veldkamp Line which is of interest to quantum information mathematicians \cite{SPPH, LHS, VL, SZ, LS}. Various properties of the Veldkamp line have been studied, in particular we are interested in the ``eloily" subgeometry formed by the canonical elliptic quadric $E_{YYY}$ (see \cite{saniga_veldkamp_GQ_2010, blunck_GQ_24_veldkamp} for more detail). This subgeometry contains no generalised triangles, and so is also known as $GQ(2,4)$, the generalised quadrangle of $4+1$ lines through each point, and $2+1$ points per line. It has $15+12=27$ points and $15+30 = 45$ lines, formed from a doily and ``double-six" arrangement, and a degree of contextuality of 9 \cite{henri_contextuality_2022}. This can be seen as it can also be composed of 3 disjoint doilies, each of which have degree of contextuality 3. A representation of the eloily is given in Fig \ref{fig:gq_24}.
\\
\begin{figure}[!h]
\centering
 \includegraphics[width=6cm]{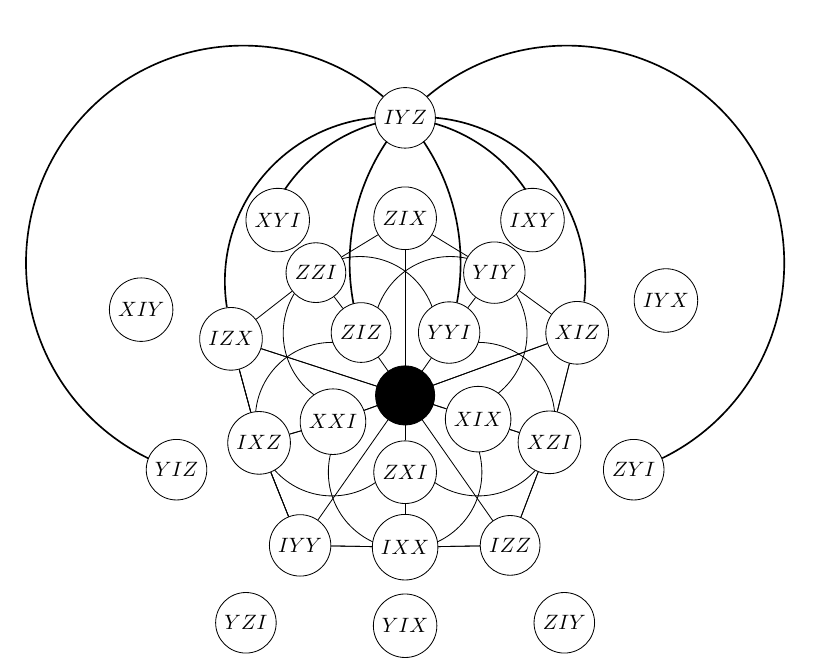}
 \includegraphics[width=7cm]{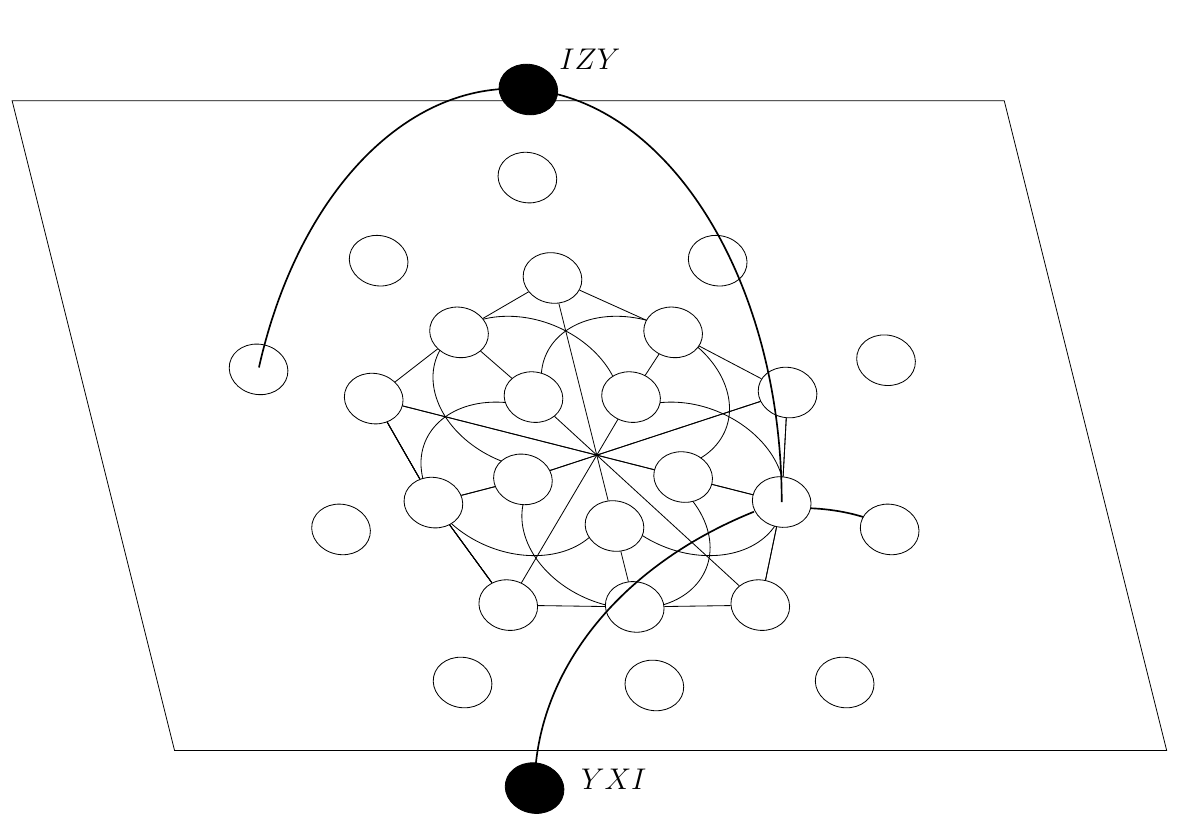}
 \caption{A visual representation of the eloily $E_{q}$, a.k.a. $GQ(2,4)$. It is formed by 15 points of a doily $D_{q}$, with a ``double-six" 12 points from $E_{q}\backslash D_{q}$. The double-six intersect the doily with representative lines as indicated in the two images. The particular qubit labelling comes from the canonical choice of $q=YYY$. For original image see Polster 1998 \cite{Polster98}. See also \cite{LSVP,blunck_GQ_24_veldkamp}.}\label{fig:gq_24}
\end{figure}

What makes the geometry of the eloily interesting is the structure of the negative lines in this canonical form. All forms of $E_{q}$ for any chosen $q$ share the property that every point is on 5 lines, and every line contains 3 points.
However in the canonical form $q=YYY$, one has the property that there are 9 non-intersecting negative lines, with every point lying on exactly one such line (referred to as a \textit{spread} of lines, see Appendix \ref{app:gq24_labelling} for details). One can use this spread to show the following property of the eloily:
\begin{proposition}\label{thm:spread}
    For any given classical strategy, each point $p$ on the eloily lies on exactly one unsatisfiable line.
\end{proposition}
\begin{proof}
    Both Charlie and the players share an ``unlabelled" geometry that indicates only line parities, such that it can be consistently labelled by quantum operators by the players. As the canonical labelling satisfies this, we consider the geometry with a spread of negative lines.
    Consider the trivial classical strategy where all points are labelled $+1$. Then the spread of negative lines is exactly the set of lines that cannot be satisfied. As the contextuality of a geometry (and thus its set of unsatisfiable lines) is dependent only on the distribution of negative lines, this shows that the set of unsatisfiable lines also forms a spread, regardless of any classical strategy that Alice and Bob may pick.
\end{proof}
The above property is unique to $GQ(2,4)$, as both the grid and doily contain points not incident with any unsatisfiable lines.
\\

\section{The Eloily Game}\label{sec:gq_24_game}

These properties allow us to form two versions of an eloily game $\mathscr{E}$: 
\begin{enumerate}[I]
    \item A standard 2-player $\mathscr{E}_{2}$ game, analogous to the Mermin \& doily games, where Charlie picks a point $p$ of $E_{YYY}$ and gives a random pair of incident lines to Alice and Bob. 
    \item A 4-player version $\mathscr{E}_{4}$ utilising the spread, where Charlie picks a point $p$ and gives one incident line each to the players Alice, Bob, Daisy and Evan.
\end{enumerate} 
In both scenarios, the players cannot determine each other's assigned lines, and must return strings of 3 elements picked from $\{+1,-1\}$. They may apply a classical strategy of populating the points of the eloily geometry with classical variables $\pm1$ as best as possible before play. However in both cases they will be left with 9 lines that cannot be satisfied, in which case they complete their responses based on their given line constraints. The winning conditions of each game are as follows, where $a_{p}$ etc. refer to the position of the player's triplet associated to the intersection point $p$:
\begin{equation}\label{eq:both_winning_condition}
    \text{Players win iff } \begin{cases}
        a_{p}\cdot b_{p} = 1, \qquad\qquad\;\, \text{I} \\
        a_{p}\cdot b_{p}\cdot d_{p}\cdot e_{p} = 1, \quad \text{II} 
    \end{cases}
\end{equation}
The 2-player condition \eqref{eq:both_winning_condition} (I) matches with that of the Mermin and doily games. The 4-player condition \eqref{eq:both_winning_condition} (II) derives from the entanglement condition on the shared state $\ket{\psi_{S,4}}$ (see Sec. \ref{sec:stabiliser_states}). For a classical labelling, provided all lines are satisfiable, both the above conditions will be met. 
\\~\\
One might wonder whether one could utilise the 5 lines through each point in the eloily geometry and form a 5-player game. This is immediately problematic due to the odd number of players. Given some winning strategy analogous to \eqref{eq:both_winning_condition}, one sees that the players must have an even number of negative measurements on the chosen intersection point $p$. But with a classical strategy they all choose the same labelling for the point $p$, and if that label were negative then the winning condition would not be met. Alternative games are likely possible to extend to odd-player games, but in this work we focus on winning conditions of the form \eqref{eq:both_winning_condition} and thus only on even-player games.
\\~\\
The eloily games $\mathscr{E}_{2}$ and $\mathscr{E}_{4}$ satisfy the following theorem:
\begin{theorem}
    The classical bounds for the eloily games are $\frac{13}{15}$ for the 2-player game $\mathscr{E}_{2}$ and $\frac{11}{15}$ for the 4-player game $\mathscr{E}_{4}$.
\end{theorem}
\begin{proof}
    For any point in the eloily and any classical strategy, there are 4 satisfiable lines and 1 unsatisfiable line intersecting at that point, per Theorem \ref{thm:spread}. Charlie wants to choose the unsatisfiable lines to maximise the odds of failure, but is blind to the specific classical strategy chosen by the players, so is unaware which lines are unsatisfiable. We will examine conditions for a failed game, and consider the standard game $\mathscr{E}_{2}$ first:
    
    Charlie has, for his chosen point, ${5\choose 2}$ ways of assigning two random lines to the players, up to swapping players. There are ${4 \choose 1}$ such choices where one player gets an unsatisfiable line, which carries with it a $\frac{1}{3}$ chance of failure as per Lemma \ref{lemma:neg_line_two_thirds}, giving a classical success rate of

    \begin{equation}
    \begin{array}{rl}
        \omega(\mathscr{E}_{2}) &= 1 - \frac{{4 \choose 1}}{{5 \choose 2}}\cdot\frac{1}{3} \\
        &= \frac{13}{15}
    \end{array}
    \end{equation}

    For the 4-player game $\mathscr{E}_{4}$, Charlie has ${5 \choose 4}$ ways of assigning 4 lines, with again ${4 \choose 1}$ such choices including an unsatisfiable line. Thus the success rate is 
    \begin{equation}
    \begin{array}{rl}
        \omega(\mathscr{E}_{4}) &= 1 - \frac{{4 \choose 1}}{{5 \choose 4}}\cdot\frac{1}{3} \\
        &= \frac{11}{15}
    \end{array}
    \end{equation}
\end{proof}

As is usual for these games, we then move to the alternative picture where we allow the players to use quantum information. Here, they share 3 pairs of entangled qubits, in such a state that they can perform measurements based on the operator-assigned points of $E_{YYY}$ and always be in agreement in their measurement outcomes. This provides the same quantum strategy success of 1 as in the Mermin and doily games. Thus in allowing 3-qubit shared entanglement between the players, and utilising the spread in $E_{YYY}$, we have constructed a game ($\mathscr{E}_{4}$) with a quantum advantage of $\frac{4}{15} = 0.2\overline{6}$.

In the next section we will describe the implementation of this game in a quantum circuit, including the construction of the shared state via stabiliser states. We begin with the 2-player $\mathcal{E}_{2}$ game before tweaking for the 4-player one.

\subsection{Implementing $\mathcal{E}_{2}$ on a Quantum Computer}\label{sec:implementation}

The implementation of this game is not immediately obvious as one must first understand what state to set the 3 entangled qubit pairs to. \\
For the Mermin-Peres and doily games with 2 qubits, the pairs of shared qubits are given by two separable Bell pairs (see \cite{Kelleher_2_qubit_games} for full discussion on implementing these two games):
\begin{equation}
    |\psi_{2Q}\rangle = \left(\frac{1}{\sqrt{2}}(|00\rangle+|11\rangle)\right)\otimes\left(\frac{1}{\sqrt{2}}(|00\rangle+|11\rangle)\right)
\end{equation}
where in each of the entangled qubit pairs in the above state, the ket $|AB\rangle$ indicates the state of Alice's qubit on the left and Bob's on the right.
\\~\\
The required (2-player) property of any initial shared state $|\psi_{2Q}\rangle$ is that any shared operator that Alice and Bob act on it give the same measurement outcome, $+1$ or $-1$. For any 2-qubit operator $\mathcal{O}_{i}$ that represents the intersection point of whatever arrangement they are playing on, we thus need
\begin{equation}\label{eq:stabiliser_states_req}
    \mathcal{O}_{i}\otimes\mathcal{O}_{i}|\psi_{2Q}\rangle = |\psi_{2Q}\rangle
\end{equation}
The above state satisfies this for even 2-qubit operators (ones with an even number of $Y$s), however for skew operators we pick up an overall phase factor of $-1$, giving
\begin{equation}
    \mathcal{O}_{\textit{skew}}\otimes\mathcal{O}_{\textit{skew}}|\psi_{2Q}\rangle = -|\psi_{2Q}\rangle
\end{equation}
Thus, the additional rule is incorporated where if Alice has skew operators in her given line, she simply flips the results of her measurements for those operators. This gives an agreed measurement value for the intersection point.

\subsection{Stabiliser States}\label{sec:stabiliser_states}
For 3 qubits, one can play without this skew-flipping rule, for both 2- and 4-player games. To do so, we make use of stabiliser states (see \cite{nielson_chuang}, Chpt. 10). We still require \eqref{eq:stabiliser_states_req} for 3 qubit operators, for some initial state $|\psi_{3Q}\rangle$.
\\~\\
To find a suitable $|\psi_{3Q}\rangle$, one makes use of the following definition of a stabiliser state:
\begin{equation}\label{eq:stabiliser_state_def}
    |\psi_{S,2}\rangle := \prod_{i}\frac{1}{2}\left(\mathbbm{1}+\mathcal{O}_{i}\otimes\mathcal{O}_{i}\right)|\chi\rangle
\end{equation}
for nontrivial 3-qubit operators $\mathcal{O}_{i}$ and any given state $|\chi\rangle$. If such a state $|\psi_{S,2}\rangle$ exists, it enjoys the property
\begin{equation}
    \mathcal{O}_{i}\otimes\mathcal{O}_{i}|\psi_{S,2}\rangle = |\psi_{S,2}\rangle
\end{equation}
for any $\mathcal{O}_{i}$, as the operators square to identity. This is the required condition on our initial state for our 2-player game, so we simply need a nontrivial $|\psi_{3Q}\rangle = |\psi_{S,2}\rangle$. 

To find a nontrivial $|\psi_{S,2}\rangle$, one can simply iterate over different $|\chi\rangle$ states and apply \eqref{eq:stabiliser_state_def} and \eqref{eq:stabiliser_state_4_player} until one of each is found. The state below is one such example:
\begin{equation}\label{eq:stabiliser_state_ex_2_player}
    |\psi_{S,2}\rangle = \left(\frac{1}{\sqrt{2}}\left(|01\rangle-|10\rangle\right)\right)\otimes\left(\frac{1}{\sqrt{2}}\left(|01\rangle-|10\rangle\right)\right)\otimes\left(\frac{1}{\sqrt{2}}\left(-|01\rangle+|10\rangle\right)\right)
\end{equation}

This state can be implemented on an IBM Qiskit\cite{ibmq,qiskit} circuit, as shown in Fig. \ref{fig:circuit_gq24}.

\begin{figure}[!h]
\centering
\includegraphics[width=0.7\textwidth]{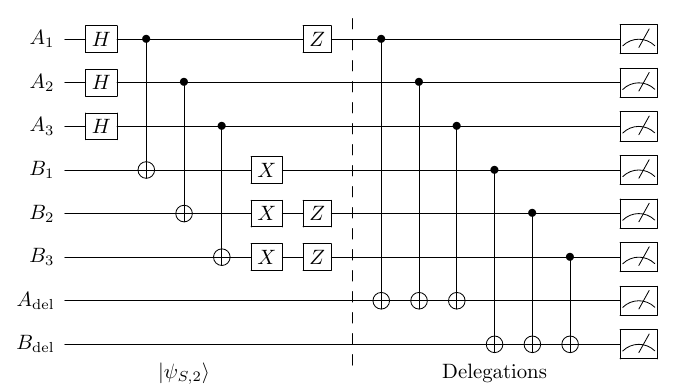}
\caption{Circuit implementing the entangled quantum state $\ket{\psi_{S,2}}$ (left), and the delegation method (right) for multiple operator measurements. Basis-change gates have been suppressed for clarity.}
\label{fig:circuit_gq24}
\end{figure}
We implement this state into a circuit and play the game. As any of the quantum games considered here involve Alice and Bob (and possibly Daisy and Evan) performing multiple measurements on an individual state, one can perform this using relevant basis transformations or via delegations onto auxiliary qubits. Here we implement this using the latter method, see \cite{Kelleher_2_qubit_games} for more details. 
\\~\\
The quantum game proceeds as follows:
\begin{enumerate}
    \item The shared state $|\psi_{S, 2}\rangle$ is implemented in the quantum circuit, with each player getting three ``register" qubits and one ``delegation" qubit (Fig. \ref{fig:circuit_gq24}).
    \item For a given point $p$ in $E_{YYY}$, Charlie selects some pair of contexts intersecting at $p$ and provides one to Alice and one to Bob.
    \item For each player, they look at the first 3-qubit operator on their context. They make relevant basis changes on their register qubits $\{A_{i},B_{i}\}$ and perform a CNOT operation targeted onto the separate delegation qubit. The state of this qubit is measured and recorded classically. Inverse basis-change gates are implemented to restore registers to original state.
    \item They perform new basis change gates based on second operators, and measure directly on their registers. They each combine the classical recordings into a single bit via bit addition.
    \item They complete their 3-bit response based on their context constraints (negative contexts must have an odd number of 1s, positive even).
    \item They send their bit strings to Charlie, who compares them. If the winning condition \eqref{eq:both_winning_condition} (I) is met, they win the game.
\end{enumerate}
The CNOT relations between registers and delegation qubits are given graphically in Fig. \ref{fig:gq24_connections}. As we need 2 extra delegation qubits, and 6 to start with, we have a total of 8 qubits in our circuit. However, based on CNOT topology, in general more qubits will be used in the physical circuit (see Fig. \ref{fig:brisbane_topology}). 

\subsection{Implementing $\mathcal{E}_{4}$ on a Quantum Circuit}
The 4-player game is very similar to $\mathcal{E}_{2}$, with differences highlighted in this section. Here, Charlie chooses a point $p$ in the geometry and provides a distinct incident line to each player Alice, Bob, Daisy and Evan. The players make measurements on a shared quantum state according to their quantum strategy, and respond to Charlie with triplets that must satisfy \eqref{eq:both_winning_condition} (II) in order to win.

First, one tweaks \eqref{eq:stabiliser_state_def} in the obvious way to find a valid condition on the shared quantum state for 4 players:
\begin{equation}\label{eq:stabiliser_state_4_player}
    \begin{array}{rl}
        |\psi_{S,4}\rangle &:= \prod_{i}\frac{1}{2}\left(\mathbbm{1}+\mathcal{O}_{i}\otimes\mathcal{O}_{i}\otimes\mathcal{O}_{i}\otimes\mathcal{O}_{i}\right)|\chi\rangle \\
        &\Rightarrow \mathcal{O}_{i}\otimes\mathcal{O}_{i}\otimes\mathcal{O}_{i}\otimes\mathcal{O}_{i}|\psi_{S,4}\rangle = |\psi_{S,4}\rangle
    \end{array}
\end{equation}
\\
To find such a $\ket{\psi_{S,4}}$, one again searches over various $\ket{\chi}$. An example of such a stabiliser state is 3 copies of the following GHZ state, entangled across the four players:
\begin{equation}\label{eq:stabiliser_state_ex_phi}
    |GHZ_{4}\rangle = \frac{1}{\sqrt{2}}\left(|0000\rangle+|1111\rangle\right)
\end{equation}
where the ket $\ket{ABDE}$ indicates that the first qubit is Alice's, etc.
\begin{equation}\label{eq:stabiliser_state_ex_4_player}
\begin{array}{rl}
    |\psi_{S,4}\rangle &= \ket{GHZ_{4}}^{\otimes 3} \\
    = \frac{1}{2\sqrt{2}}&(\ket{000000000000} + \ket{001001001001}
    + \ket{010010010010} + \ket{011011011011} \\
    &+ \ket{100100100100} + \ket{101101101101}
    + \ket{110110110110} + \ket{111111111111})
    \end{array}
\end{equation}
where now the first three qubits are assigned to Alice, the next 3 to Bob, etc. 
\\~\\
To implement this into a quantum circuit in analogy with Fig. \ref{fig:circuit_gq24}, one puts Alice's register qubits into superposition with $H$ gates, as in the $\mathcal{E}_{2}$ example, and then applies CNOTs between Alice's register and each of Bob's, Daisy's and Evan's to entangle qubits 1, 2 and 3 across the 4 players.

Then, as per the $\mathcal{E}_{2}$ game, delegations are made to a delegation qubit per player, based on the first operator in each player's context. These delegation qubits are measured and the results stored classically to give each players first triplet entry. Basis changes are made to the register qubits and direct measurements are made on those according to the second operator on each player context, and combined to give the second triplet entries. The third are computed from the context parities. The players then win iff \eqref{eq:both_winning_condition} (II) is satisfied.

\section{Results \& Conclusion}\label{sec:results}
% \begin{itemize}
%     \item Needs more material
%     \item Put in table with simulated, sim with noise, and actual results
%     \item Conclusion?
% \end{itemize}
This circuit was implemented on the 127-qubit \texttt{ibm\_brisbane} backend, and looped over all points in $E_{YYY}$ and all permitted line choices for the players. Each circuit was run with a total $8192$ shots, giving $8192 \times 27 \times 10 = 2,211,840$ data points for $\mathcal{E}_{2}$ and $8192 \times 27 \times 5 = 1,105,920$ for $\mathcal{E}_{4}$. The results are given in Table \ref{tab:gq_results}, with simulated results for both and real quantum backend results for $\mathcal{E}_{2}$. The $E_{YYY}$ geometry also contains $120$ grids and $36$ doilies as subgeometries. The results for playing associated games were extracted from the $\mathscr{E}_{2}$ results data, and the best performing grid $\mathscr{M}$ and doily $\mathscr{D}$ results are also shown in Table \ref{tab:gq_results}. These subgeometry results are an improvement upon the results from \cite{Kelleher_2_qubit_games}. The particular grid and doily are shown in Fig. \ref{fig:grid_doily_best}.
\\

\begin{table}[h!]
\centering
\begin{tabularx}{0.781\textwidth}{|c||ccc||l|} \hline 
     \qquad \newline \text{Game $\mathscr{G}$} & \qquad \newline Noiseless Simulation & Noisy Simulation &  \qquad \newline $\sigma(\mathscr{G})$ & \qquad \newline $\omega(\mathscr{G})$ \\ \hline \hline
     $\mathscr{M}$ & 1.0 & $0.93490$ & ${0.87904}$ & $0.\overline{8}$ \\
     $\mathscr{D}$ & 1.0 & $0.93891$ & ${0.84230}$ & $0.8\overline{6}$ \\
    $\mathscr{E}_{2}$ & $1.0$ & $0.94559$ & ${0.82137}$ & $0.8\overline{6}$ \\ 
    $\mathscr{E}_{4}$ & $1.0$ & $0.83726$ & $\sim0.5$ & $0.7\overline{3}$ \\ \hline
    
\end{tabularx}
\caption{Results for computing the average success rate $\sigma(\mathscr{G})$ on the IBM ``Brisbane" NISQ Computer for the grid ($\mathscr{M}$), doily ($\mathscr{D}$) and eloily ($\mathscr{E}_{2}, \mathscr{E}_{4}$) games. Shown alongside are classical bounds $\omega(\mathscr{G})$ and results for simulators with and without noise models, however note that the noise models do not consider circuit transpilations due to CNOT topology (see Section \ref{sec:transpilation_noise}).} 
\label{tab:gq_results}
\end{table}

\begin{figure}[!h]
\centering
\begin{subfigure}[]{.4\textwidth}
  \centering
  \includegraphics[width=0.7\linewidth]{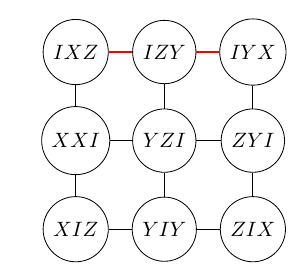}
  \label{fig:grid_best}
\end{subfigure}
\qquad
\begin{subfigure}[]{.4\textwidth}
  \centering
 \includegraphics[width=0.9\linewidth]{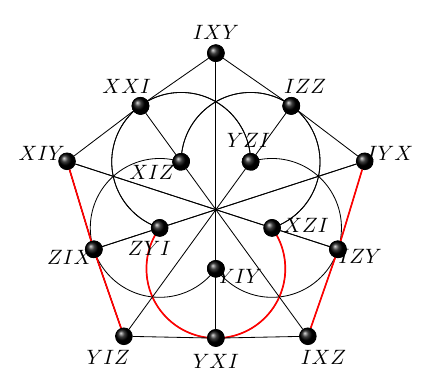}
  \label{fig:doily_best}
  \end{subfigure}
 \caption{Best performing grid (left) and doily (right) subgeometries from the eloily results. Red lines indicate negative contexts.}\label{fig:grid_doily_best}
\end{figure}

Alice and Bob clearly beat the bound in all simulations, both with and without noise models. However, on the actual quantum backends they were unsuccessful in beating the classical bounds. This is due to transpilation-induced noise arising from the circuit adding additional CNOT gates in order to implement CNOTs between qubits not directly connected in the topology (see Sec. \ref{sec:transpilation_noise}). The ``smaller" games of the Mermin grid and doily come closest to beating the bound, which is unsurprising as there are many copies of these geometries within $GQ(2,4)$, and so the best performing copies can be highlighted. It is interesting to note that the 2-player eloily game has the same classical bound as the doily game, likely from the fact that each intersecting pair of lines in $GQ(2,4)$ live in some doily as a subgeometry. Finally, it is worth pointing out that the spread of unsatisfiable lines has afforded us a lower bound by using 4 players and increasing the chances of Charlie choosing such a line. An experimental demonstration of the quantum advantage was not shown for the $\mathcal{E}_{4}$ game only due to the CNOT topology of the IBM quantum computers (Sec. \ref{sec:transpilation_noise}), however were some other backend to become available with a more interconnected qubit structure, such a game could be used to demonstrate the lack of NCHV models underlying the system.

\subsection{Transpilation-Induced Noise}\label{sec:transpilation_noise}

Noise models in the IBM Quantum Experience ignore topology-induced transpilation, and any additional CNOT gates introduced thereof. Thus the noisy simulator results differ from the actual results $\sigma(\mathscr{G})$ in Table \ref{tab:gq_results}.

Because of the CNOT gates required between both Alice and the other player's qubits, and between these qubits and the delegation qubits (Fig. \ref{fig:circuit_gq24}), it is not possible to design the circuit such that every CNOT is between directly connected qubits in the topology (see Fig. \ref{fig:brisbane_topology_and_mapping}). A mapping was chosen that relied on additional qubits symmetrically for both players, however possible other mappings may provide an even lower induced noise level. In particular, the mapping used in this work incurred a high ``circuit qubit-distance" between connected qubits $A_{2}$, $B_{2}$ and between $A_{\text{del}}$, $A_{3}$, and Bob's counterparts.
\begin{figure}[!h]
\centering
\begin{subfigure}[b]{.4\textwidth}
  \centering
  \includegraphics[width=\linewidth]{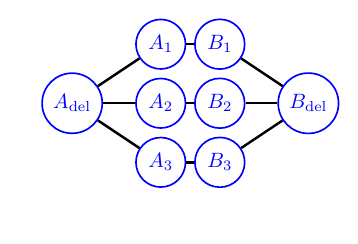}
  \caption{$\mathcal{E}_{2}$ game CNOT connections.}
  \label{fig:gq24_connections}
\end{subfigure}
\qquad
\begin{subfigure}[b]{.4\textwidth}
  \centering
 \includegraphics[width=\linewidth]{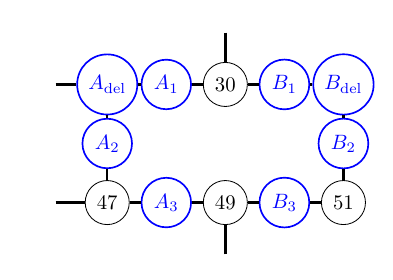}
 \caption{IBM ``Brisbane" CNOT topology.}
  \label{fig:brisbane_topology}
  \end{subfigure}
 \caption{Mapping the required qubit CNOT connections for the 2-player eloily game (left) onto a section of the IBM ``Brisbane" topology (right). Nodes represent qubits, and edges represent direct CNOT connections. Alice is assigned qubits $\{A_{1},A_{2},A_{3}\} = \{29, 35, 48\}$ and delegation qubit $A_{\text{del}} = 28$. Bob is assigned qubits $\{B_{1},B_{2},B_{3}\} = \{31,36,50\}$ and delegation qubit $B_{\text{del}}=32$, all indicated in blue. Qubits $\{30,47,49,51\}$ in black are auxiliary qubits used for additional CNOT gates in transpiled circuit. For 4-player game, additional players Daisy and Evan are given register and delegation qubits, connected to Alice's in the same manner as Bob's.}\label{fig:brisbane_topology_and_mapping}
\end{figure}

\section{Three-qubit invariants and the Cabello inequality}\label{sec:quadrangles}
We finish this article with a surprising connection between a well-known inequality quantifying quantum contextuality, and an invariant of the geometries described here. The point-line incidence structures we presented in this paper are all generalized quadrangles\cite{distance_graphs}:
\begin{definition}
 A \textit{generalized quadrangle} $GQ(s,t)$ is a point-line geometry with the property that any line contains $s+1$ points and any point belongs to $t+1$ lines and such that the configuration is triangle-free.
\end{definition}
\begin{example}
The Mermin-square has the incidence structure of a grid, i.e. $GQ(2,1)$, the doily corresponds to $GQ(2,2)$ and the eloily is $GQ(2,4)$ as previously stated. These are the only nontrivial quadrangles for $s=2$.
\end{example}

The canonical labelling chooses the 27 points of $\mathcal{W}(5,2)$ that contain exactly one identity operator (see App. \ref{app:gq24_labelling}), and through the position of the identity can be split in to three disjoint Mermin grids, given below in matrix form:
% \newpage
\\~\\
\begin{equation}\label{eq:gq24_subgrids}
    \mathbb{A} = 
    \begin{bmatrix}
    XYI&ZXI&YZI\\
    ZZI&YYI&XXI\\
    YXI&XZI&ZYI
    \end{bmatrix}, \;
    \mathbb{B} = 
    \begin{bmatrix}
    YIX&XIZ&ZIY\\
    ZIZ&YIY&XIX\\
    XIY&ZIX&YIZ
    \end{bmatrix}, \;
    \mathbb{C} = 
    \begin{bmatrix}
    IXY&IZX&IYZ\\
    IZZ&IYY&IXX\\
    IYX&IXZ&IZY
    \end{bmatrix}
\end{equation}

Now, in general one can form the following 3-qubit Hermitian observable
\begin{equation}\label{eq:hermitian_H}
    \mathscr{H} := \mathscr{A} + \mathscr{B} + \mathscr{C}
\end{equation}
with $\mathscr{A}, \mathscr{B}, \mathscr{C}$ each collections of 2-qubit operators and an identity:
\begin{equation}\label{eq:operator_A}
    \mathscr{A} := \sum_{ij}A^{ij}\sigma_{i}\otimes\sigma_{j}\otimes I, \quad \mathscr{B} := \sum_{ij}B^{ij}\sigma_{i}\otimes I \otimes \sigma_{j}, \quad \mathscr{C} := \sum_{ij}C^{ij}I \otimes \sigma_{i} \otimes \sigma_{j}
\end{equation}
Here the coefficients $A^{ij}$ etc. take on the values $\pm1$, and so for each grid $\mathbb{A},\mathbb{B},\mathbb{C}$, there is an associated matrix $A,B,C$ of classical hidden variables.

We are interested in the following cubic invariant (see \cite{faulkner} for original discussion, and \cite{LSVP} for connection with black hole entropy):
\begin{equation}\label{eq:cartan_invariant}
    I_{3}(\mathscr{H}) := -\frac{1}{48}\text{Tr}(\mathscr{H}^{3})
\end{equation}
Writing this in terms of the separated components \eqref{eq:hermitian_H} and their coefficient matrices one obtains:
\begin{equation}\label{eq:cartan_det_form}
    I_{3}(\mathscr{H}) = \text{Det}(A) + \text{Det}(B) + \text{Det}(C) - \text{Tr}(CB^{T}A)
\end{equation}
This invariant describes the point-line incidence structure of $GQ(2,4)$, the eloily, in the sense that it contains 27 parameters spread out over 45 terms. Each term is associated with a context, and the parameters with the collinear points. Moreover the monomials come with signs corresponding to the distribution of signs induced by the labeling of the configuration by $n$-qubit Pauli matrices. When the parameters are given by classical variables $A^{ij}$ etc., the invariant expresses the evaluation of all contexts by a HV model along with the line parities (term coefficients of $\pm1$). It can be reduced to invariants over the grid and doily by eliminating certain terms, briefly described at the end of this Section.

Related to this is the celebrated Cabello Inequality\cite{cabello_proposed_2010} featuring the following measure:
\begin{equation}
    \chi := \sum_{i}\langle C_{i} \rangle - \sum_{j}\langle C'_{j} \rangle
\end{equation}
where $\langle C \rangle$ is the expectation value of measuring along a positive context and $\langle C' \rangle$ along a negative context. Under the hypothesis of QM on has $\chi=N$ where $N$ is the number of contexts. However if one consider a HV model with valuation on the node defined by the matrices $A, B$ and $C$ one has:
\begin{equation}\label{eq:I3_chi}
    I_{3}(\mathscr{H}) = \chi
\end{equation}

In the notation of Cabello, one has the following testable inequality for contextuality:
\begin{equation}\label{eq:cabello_ineq}
    \chi \leq  \begin{cases}
        N, \quad \quad \quad \text{QM} \\
        N - 2d, \quad \text{hidden variables}
    \end{cases} 
\end{equation}
where $d$ its degree of contextuality.

One can sanity check the relation between $\chi$ and $I_{3}(\mathscr{H})$ via reducing to the $GQ(2,1)$ case, where $B = C = 0$ in \eqref{eq:cartan_det_form}. Here $I_{3}(\mathscr{H}) = \text{Det}(A)$, with $A$ a $3\times 3$ matrix with entries $\pm1$. As is known\cite{Hadamard_matrices} about such ``Hadamard" matrices, the determinant is valued at either $0$ or $4$. The upper bound for the HV model in \eqref{eq:cabello_ineq} gives $6-2\cdot 1=4$ as expected.

\begin{table}[!h]
    \centering
    \begin{tabular}{c|cc|cc|c}
     & $N$& $d$ & $N-2d$ & max($I_{3}$) & $t$ (sec) \\ \hline
    $GQ(2,1)$ & $6$ & $1$ & $4$ & $4$ & $<0.008$ \\ 
    $GQ(2,2)$ & $15$ & $3$ & $9$ & $9$ & $<0.119$ \\ 
    $GQ(2,4)$ & $45$ & $9$ & $27$ & $27$ & $<772$ \\ \hline
    \end{tabular}
    \caption{Table showing computed maximal values of $I_{3}$ compared to HV bound in Cabello inequality $N - 2d$, for the three quadrangles of lines size 3. Indicated on right is compute time in seconds for comparing all HV inputs.}
    \label{tab:invariant_computations}
\end{table}

For our quadrangles $GQ(2,1), GQ(2,2), GQ(2,4)$ we have degrees of contextuality $1, 3, 9$ respectively, and computed maximal values of $I_{3}(\mathscr{H})$ across all possible HV assignments are given in Table \ref{tab:invariant_computations}. For the doily case, one can express $GQ(2,4)$ as a disjoint doily and double-six arrangement, as has been done in this work, and discard the double-six component in $I_{3}(\mathscr{H})$. Then the invariant becomes the Pfaffian of a $6\times 6$ antisymmetric matrix with HV entries. The results show that they agree with the bound given by \eqref{eq:cabello_ineq}, confirming that the contextuality of these quadrangles is encapsulated by this cubic invariant.
\linebreak

\section*{Acknowledgments}
This work is supported by the Graduate school EIPHI (contract ANR-17-EURE- 0002) through the project TACTICQ, the Ministry of Culture and Innovation and the National Research, Development and Innovation Office within the Quantum Information National Laboratory of Hungary (Grant No. 2022-2.1.1-NL-2022-00004) and by the National Research Development and Innovation Office of Hungary within the Quantum Technology National Excellence Program (Project No. 2017-1.2.1-NKP-2017-0001).  We acknowledge the use of the IBM Quantum Experience and the IBMQ-research program. The views expressed are those of the authors and do not reflect the official policy or position of IBM or the IBM Quantum Experience team. One would like to thank the developers of the open-source framework Qiskit.

Péter Lévay would like to express his gratitude for the warm hospitality at UTBM during his visit. The participation of Péter in this collaboration has been supported by a visiting professor grant awarded in 2023. The authors also want to thank Metod Saniga of the Astronomical Institute Slovak Academy of Sciences for coining the term ``eloily".

\printbibliography

@book{nielson_chuang,
  Author = {Michael A. Nielsen and Isaac L. Chuang},
  Title = {Quantum Computation and Quantum Information: 10th Anniversary Edition},
  Publisher = {Cambridge University Press},
  Year = {2011},
  ISBN = {9781107002173},
  doi = {https://doi.org/10.1017/CBO9780511976667}
}

@article{bell1964einstein,
  title={On the einstein podolsky rosen paradox},
  author={Bell, John S},
  journal={Physics Physique Fizika},
  volume={1},
  number={3},
  pages={195},
  year={1964},
  publisher={APS}
}

@article{bell_problem_1966,
	author = {BELL, JOHN S.},
	doi = {10.1103/RevModPhys.38.447},
	journal = {Reviews of Modern Physics},
	month = jul,
	note = {Publisher: American Physical Society},
	number = {3},
	pages = {447--452},
	title = {On the {Problem} of {Hidden} {Variables} in {Quantum} {Mechanics}},
	url = {https://link.aps.org/doi/10.1103/RevModPhys.38.447},
	urldate = {2023-07-05},
	volume = {38},
	year = {1966}

}

@article{Kelleher_2_qubit_games,
	author = {Kelleher, Colm and Roomy, Mohammed and Holweck, Fr\'ed\'eric},
	doi = {arXiv:2310.07441},
	journal = {arxiv},
	month = oct,
	note = {arXiv:2310.07441},
	title = {Implementing 2-qubit pseudo-telepathy games on noisy intermediate scale quantum computers},
	url = {arXiv:2310.07441}

}

@article{xu_experimental_2022,
	author = {Xu, Jia-Min and Zhen, Yi-Zheng and Yang, Yu-Xiang and Cheng, Zi-Mo and Ren, Zhi-Cheng and Chen, Kai and Wang, Xi-Lin and Wang, Hui-Tian},
	doi = {10.1103/PhysRevLett.129.050402},
	journal = {Physical Review Letters},
	month = jul,
	note = {Publisher: American Physical Society},
	number = {5},
	pages = {050402},
	title = {Experimental {Demonstration} of {Quantum} {Pseudotelepathy}},
	url = {https://link.aps.org/doi/10.1103/PhysRevLett.129.050402},
	urldate = {2023-02-08},
	volume = {129},
	year = {2022}}

@article{cabello_proposed_2010,
	author = {Cabello, Ad{\'a}n},
	doi = {10.1103/PhysRevA.82.032110},
	journal = {Physical Review A},
	month = sep,
	note = {Publisher: American Physical Society},
	number = {3},
	pages = {032110},
	title = {Proposed test of macroscopic quantum contextuality},
	url = {https://link.aps.org/doi/10.1103/PhysRevA.82.032110},
	urldate = {2023-02-08},
	volume = {82},
	year = {2010}}

@article{aravind_bells_2002,
	author = {Aravind, P. K.},
	doi = {10.1023/A:1021272729475},
	issn = {1572-9524},
	journal = {Journal of Genetic Counseling},
	language = {en},
	month = aug,
	number = {4},
	pages = {397--405},
	title = {Bell's {Theorem} {Without} {Inequalities} and {Only} {Two} {Distant} {Observers}},
	url = {https://doi.org/10.1023/A:1021272729475},
	urldate = {2023-02-08},
	volume = {15},
	year = {2002}}

@article{brassard_quantum_2005,
	author = {Brassard, Gilles and Broadbent, Anne and Tapp, Alain},
	doi = {10.1007/s10701-005-7353-4},
	issn = {1572-9516},
	journal = {Foundations of Physics},
	language = {en},
	month = nov,
	number = {11},
	pages = {1877--1907},
	title = {Quantum {Pseudo}-{Telepathy}},
	url = {https://doi.org/10.1007/s10701-005-7353-4},
	urldate = {2023-02-08},
	volume = {35},
	year = {2005}}

@article{kelleher_x_states_2021,
	author = {Kelleher, Colm and Holweck, Fr{\'e}d{\'e}ric and L{\'e}vay, P{\'e}ter and Saniga, Metod},
	copyright = {All rights reserved},
	doi = {10.1016/j.rinp.2021.103859},
	issn = {2211-3797},
	journal = {Results in Physics},
	language = {en},
	month = mar,
	pages = {103859},
	title = {X-states from a finite geometric perspective},
	url = {https://www.sciencedirect.com/science/article/pii/S2211379721000425},
	urldate = {2023-02-08},
	volume = {22},
	year = {2021}}

@inbook{Polster98,
  author    = {Polster, Burkand},
  title     = {A Geometrical Picture Book},
  chapter   = {Generalized Quadrangles},
  pages     = {39-66},
  publisher = {Springer},
  address   = {New York, NY},
  year      = {1998},
  url       = {https://doi.org/10.1007/978-1-4419-8526-2_4}
}

@article{Budroni21,
  title={Quantum contextuality},
  author={Budroni, C. and Cabello, A. and G{\"u}hne, O. and Kleinmann, M. and Larsson, J. {\AA}.},
  journal={arXiv preprint arXiv:2102.13036},
  year={2021}
}

@article{Mermin93,
  title={Hidden variables and the two theorems of John Bell},
  author={Mermin, N. David},
  journal={Reviews of Modern Physics},
  volume={65},
  number={3},
  pages={803},
  year={1993},
  publisher={APS}
}

@article{Peres90,
  title={Incompatible results of quantum measurements},
  author={Peres, Asher},
  journal={Physics Letters A},
  volume={151},
  number={3-4},
  pages={107--108},
  year={1990},
  publisher={Elsevier}
}

@incollection{Kochen_specker,
  title={The problem of hidden variables in quantum mechanics},
  author={Kochen, Simon and Specker, Ernst P.},
  booktitle={The logico-algebraic approach to quantum mechanics 1975},
  pages={293--328},
  year={1975},
  publisher={Springer}
}

@article{Aravind04,
  title={Quantum mysteries revisited again},
  author={Aravind, P. K.},
  journal={American Journal of Physics},
  volume={72},
  number={10},
  pages={1303--1307},
  year={2004},
  publisher={AIP Publishing}
}

@article{Brassard05,
  title={Quantum pseudo-telepathy},
  author={Brassard, Gilles and Broadbent, Anne and Tapp, Alain},
  journal={Foundations of Physics},
  volume={35},
  number={11},
  pages={1877--1907},
  year={2005},
  publisher={Springer}
}

@article{Arkhipov12,
  title={Extending and characterizing quantum magic games},
  author={Arkhipov, Alex},
  journal={arXiv preprint arXiv:1209.3819},
  year={2012}
}

@online{ibmq,
  title={IBM Quantum},
  url={https://quantum-computing.ibm.com/},
  urldate={2023-03-15}
}

@article{Holweck21,
  title={Testing quantum contextuality of binary symplectic polar spaces on a Noisy Intermediate Scale Quantum Computer},
  author={Holweck, Frédéric},
  journal={Quantum Information Processing},
  volume={20},
  number={7},
  pages={1--13},
  year={2021},
  publisher={Springer}
}

@article{DRLB,
  title={Measuring the Mermin-Peres magic square using an online quantum computer},
  author={Dikme, Ali and Reichel, Norbert and Laghaout, Amine and Björk, Gunnar},
  journal={arXiv preprint arXiv:2009.10751},
  year={2020}
}

@article{T,
  title={The geometry of generalized Pauli operators of N-qudit Hilbert space, and an application to MUBs},
  author={Thas, Koen},
  journal={EPL (Europhysics Letters)},
  volume={86},
  number={6},
  pages={60005},
  year={2009},
  publisher={IOP Publishing}
}

@online{saniga_doily_gem,
author = {Saniga, Metod},
year = {2019},
month = {06},
pages = {},
title = {Doily - A Gem of the Quantum Universe},
organisation = {Slovenian International Conference on Graph Theory},
url = {https://www.astro.sk/~msaniga/pub/ftp/bled.pdf}
}

@article{LSVP,
  title={Black hole entropy and finite geometry},
  author={L{\'e}vay, P. and Saniga, M. and Vrana, P. and Pracna, P.},
  journal={Physical Review D},
  volume={79},
  number={8},
  pages={084036},
  year={2009},
  publisher={APS}
}

@article{SPPH,
  title={The Veldkamp space of two-qubits},
  author={Saniga, M. and Planat, M. and Pracna, P. and Havlicek, H.},
  journal={SIGMA. Symmetry, Integrability and Geometry: Methods and Applications},
  volume={3},
  pages={075},
  year={2007},
  publisher={Mathematical Institute, Slovak Academy of Sciences}
}

@article{LHS,
  title={Magic three-qubit Veldkamp line: A finite geometric underpinning for form theories of gravity and black hole entropy},
  author={L{\'e}vay, P{\'e}ter and Holweck, Fr\'ed\'eric and Saniga, Metod},
  journal={Physical Review D},
  volume={96},
  number={2},
  pages={026018},
  year={2017},
  publisher={APS}
}

@article{VL,
  title={The Veldkamp space of multiple qubits},
  author={Vrana, Peter and L{\'e}vay, P{\'e}ter},
  journal={Journal of Physics A: Mathematical and Theoretical},
  volume={43},
  number={12},
  pages={125303},
  year={2010},
  publisher={IOP Publishing}
}

@article{SZ,
  title={Magic Three-Qubit Veldkamp Line and Veldkamp Space of the Doily},
  author={Saniga, Metod and Szab{\'o}, Zsolt},
  journal={Symmetry},
  volume={12},
  number={6},
  pages={963},
  year={2020},
  publisher={MDPI}
}

@article{Hadamard_matrices,
  title={Embedding and Extension Properties of Hadamard Matrices Revisited},
  author={Christou, Dimitrios and Mitrouli, Marilena and Seberry, Jennifer},
  journal={Special Matrices},
  volume={6},
  number={1},
  pages={155-165},
  year={2018},
  doi = {10.1515/spma-2018-0012}
}

@article{henri_contextuality_2022,
	author = {de Boutray, Henri and Holweck, Fr{\'e}d{\'e}ric and Giorgetti, Alain and Masson, Pierre-Alain and Saniga, Metod},
	doi = {10.1088/1751-8121/aca36f},
	issn = {1751-8113, 1751-8121},
	journal = {Journal of Physics A: Mathematical and Theoretical},
	language = {en},
	month = nov,
	note = {arXiv:2105.13798 [math-ph, physics:quant-ph]},
	number = {47},
	pages = {475301},
	title = {Contextuality degree of quadrics in multi-qubit symplectic polar spaces},
	url = {http://arxiv.org/abs/2105.13798},
	urldate = {2023-07-13},
	volume = {55},
	year = {2022}}

@article{muller_multi_qubit_2022,
	author = {Muller, Axel and Saniga, Metod and Giorgetti, Alain and De Boutray, Henri and Holweck, Fr{\'e}d{\'e}ric},
	doi = {10.1016/j.jocs.2022.101853},
	issn = {18777503},
	journal = {Journal of Computational Science},
	language = {en},
	month = oct,
	note = {arXiv:2206.03599 [quant-ph]},
	pages = {101853},
	shorttitle = {Multi-qubit doilies},
	title = {Multi-qubit doilies: enumeration for all ranks and classification for ranks four and five},
	url = {http://arxiv.org/abs/2206.03599},
	urldate = {2023-07-13},
	volume = {64},
	year = {2022}}

@article{blunck_GQ_24_veldkamp,
	author = {Blunck, Andrea and L{\'e}vay, P{\'e}ter and Saniga, Metod and Vrana, P{\'e}ter},
	doi = {10.1080/03081087.2011.651725},
	issn = {0308-1087, 1563-5139},
	journal = {Linear and Multilinear Algebra},
	language = {en},
	month = oct,
	number = {10},
	pages = {1143--1154},
	title = {Invertible symmetric 3 × 3 binary matrices and {GQ}(2, 4)},
	url = {http://www.tandfonline.com/doi/abs/10.1080/03081087.2011.651725},
	urldate = {2023-07-11},
	volume = {60},
	year = {2012}}

@article{saniga_veldkamp_GQ_2010,
	author = {Saniga, M. and Green, R. M. and L{\'e}vay, P. and Vrana, P. and Pracna, P.},
	doi = {10.1142/S0219887810004762},
	issn = {0219-8878, 1793-6977},
	journal = {International Journal of Geometric Methods in Modern Physics},
	language = {en},
	month = nov,
	number = {07},
	pages = {1133--1145},
	title = {{THE} {VELDKAMP} {SPACE} {OF} {GQ}(2, 4)},
	url = {https://www.worldscientific.com/doi/abs/10.1142/S0219887810004762},
	urldate = {2023-07-11},
	volume = {07},
	year = {2010}}

@article{faulkner,
    title={A construction of Lie algebras from a class of ternary algebras},
    author={Faulkner, John R.},
    journal={Trans. American Mathematical Society},
    volume={155},
    number={2},
    pages={397–408},
    year={1971},
    doi = {10.1090/S0002-9947-1971-0294424-X}
    
}

@article{LS,
  title={Mermin pentagrams arising from Veldkamp lines for three qubits},
  author={L{\'e}vay, P{\'e}ter and Szab{\'o}, Zsolt},
  journal={Journal of Physics A: Mathematical and Theoretical},
  volume={50},
  number={9},
  pages={095201},
  year={2017},
  publisher={IOP Publishing}
}

@misc{qiskit,
  title={Qiskit: An open-source framework for quantum computing},
  author={{Qiskit Development Team}},
  year={2021},
  howpublished={\url{https://qiskit.org/}},
  note={Accessed: 2021-09-29}
}

@inbook{distance_graphs,
    author = {Brouwer, A. E. and Cohen, A. M. and Neumaier, A.},
    title = {Distance Regular Graphs},
    publisher = {Springer},
    year = {1989},
    chapter = {1},
    pages     = {29--33},
doi = {https://doi.org/10.1007/978-3-642-74341-2}
}
\pagebreak
\appendix
\section{Canonical Labelling of the Eloily}\label{app:gq24_labelling}

{\centering
{\footnotesize
    \begin{tabular}{cc|cc}%
        Subgeometry & Points & Lines & Line Parity \\ \hline
        \multirow{15}{*}{Doily $\qquad \; \left\lbrace\begin{array}{r}
            \\
            \\
            \\
            \\
            \\
            \\
            \\
            \\
            \\
            \\
            \\
            \\
            \\
            \\
            \\
              
        \end{array}\right.$} & IXZ & ZIZ - ZXI - IXZ & +1 \\
    	& XIZ & ZIZ - ZZI - IZZ & +1 \\
    	& XXI & IYY - IZX - IXZ & +1 \\
    	& ZXI & IYY - YYI - YIY & +1 \\
    	& ZZI & XXI - XIZ - IXZ & +1 \\
    	& IZX & XXI - XIX - IXX & +1 \\
    	& IXX & ZIX - ZXI - IXX & +1 \\
    	& ZIZ & ZIX - XIZ - YIY & +1 \\
    	& ZIX & ZIX - ZZI - IZX & +1 \\
    	& IZZ & XZI - ZXI - YYI & +1 \\
    	& YYI & XZI - XIX - IZX & +1 \\
    	& XIX & XZI - XIZ - IZZ & +1 \\ 
    	& XZI & XXI - ZZI - YYI & -1 \\
    	& YIY & IYY - IXX - IZZ & -1 \\
    	& IYY & ZIZ - XIX - YIY & -1 \\ 
     \multirow{12}{*}{Double-Six $\left\lbrace\begin{array}{r}
            \\
            \\
            \\
            \\
            \\
            \\
            \\
            \\
            \\
            \\
            \\
            \\
              
        \end{array}\right.$} & ZYI & ZYI - IYZ - ZIZ & +1 \\
 	& YIX & ZYI - ZIY - IYY & +1 \\
 	& XYI & ZYI - YZI - XXI & +1 \\
 	& IXY & ZYI - IYX - ZIX & +1 \\
 	& YIZ & YIX - XIY - ZIZ & +1 \\
 	& IZY & YIX - YZI - IZX & +1 \\
 	& XIY & YIX - IYX - YYI & +1 \\
 	& IYZ & YIX - YXI - IXX & +1 \\
 	& ZIY & XYI - XIY - IYY & +1 \\
 	& YZI & XYI - IYZ - XIZ & +1 \\
 	& IYX & XYI - IYX - XIX & +1 \\
 	& YXI & XYI - YXI - ZZI & +1 \\
 	& & IXY - XIY - XXI & +1 \\
 	& & IXY - ZIY - ZXI & +1 \\
 	& & IXY - IYX - IZZ & +1 \\
 	& & IXY - YXI - YIY & +1 \\
 	& & YIZ - IYZ - YYI & +1 \\
 	& & YIZ - ZIY - XIX & +1 \\
 	& & YIZ - YZI - IZZ & +1 \\
 	& & YIZ - YXI - IXZ & +1 \\
 	& & IZY - XIY - XZI & +1 \\
 	& & IZY - IYZ - IXX & +1 \\
 	& & IZY - ZIY - ZZI & +1 \\
 	& & IZY - YZI - YIY & +1 \\
 	& & IZY - IYX - IXZ & -1 \\
 	& & YIZ - XIY - ZIX & -1 \\
 	& & IXY - IYZ - IZX & -1 \\
 	& & XYI - YZI - ZXI & -1 \\
 	& & YIX - ZIY - XIZ & -1 \\
 	& & ZYI - YXI - XZI & -1 \\
    \end{tabular}
    }
\newline
\\
The canonical labelling of $GQ(2,4)$, given by $E_{YYY}$. All points and lines labelled by ``Doily" are contained in $D_{YYY}$, while the ``Double-Six" points are contained in the complement. All lines in the ``Double-Six" component intersect exactly one point of $D_{YYY}$.
}

\end{document}